\documentclass[reqno]{amsart}
\usepackage[T1]{fontenc}
\usepackage[latin1]{inputenc}
\usepackage{amsmath}
\usepackage{amssymb}
\usepackage{amsfonts}
\usepackage{amsthm}
\usepackage{graphicx}
\usepackage{enumerate}
\usepackage{hyperref}

\newcommand{\beq}{\begin{eqnarray}}
\newcommand{\eeq}{\end{eqnarray}}
\newcommand{\nn}{\nonumber}

\newtheorem{thm}{Theorem}[section]

\newtheorem{defn}[thm]{Definition}
\newtheorem{prop}[thm]{Proposition}
\newtheorem{rem}[thm]{Remark}
\newtheorem{lem}[thm]{Lemma}

\begin{document}

\title{The principle of stationary action in the calculus of variations}
\author{E. L\'opez, A. Molgado, J. A. Vallejo}
\address{Facultad de Ciencias\\
 Universidad Aut\'onoma de San Luis Potos\'{\i} (M\'exico)\\
 Lateral Av. Salvador Nava s/n, SLP 78290.}
\email{emanuellc@uaslp.edu.mx,molgado@fc.uaslp.mx,jvallejo@fc.uaslp.mx}

\begin{abstract}
We review some techniques from non-linear analysis
in order to investigate critical paths for the action functional
in the calculus of variations applied to physics. Our main intention
in this regard is to  
expose precise mathematical conditions for critical paths 
to be minimum solutions in a variety of situations of interest in Physics.  Our claim is that,
with a few elementary techniques,
a systematic analysis (including the  domain for which critical points are genuine minima)
of non-trivial models is possible.
We present specific models arising in modern physical 
theories in order to make clear the ideas here exposed.
\end{abstract}

\maketitle

\section{Introduction}

The calculus of variations is one of the oldest techniques of differential calculus. Ever since its creation
by Johann and Jakob Bernoulli in 1696-97, to solve the problem of the brachistochrone (others solved it, too: Newton,
Leibniz, Tschirnhaus and L'H\^opital, but their methods were different), it has been applied to a variety of problems
both in pure and applied mathematics. While occupying a central place in modern engineering techniques (mainly in control theory, see \cite{burghes}, \cite{hermes}, \cite{lebedev}, \cite{sussmann}, \cite{zaslavski}), it
is in physics where its use has been promoted to the highest level, that of the basic principle to obtain the equations
of motion, both in the dynamics of particles and fields: the principle of stationary action (see \cite{basdevant},
\cite{carlini2}, \cite{novikov}, \cite{ramond}). Accordingly
to that point of view, almost every text on mechanics include a chapter on the calculus of variations although,
surprisingly enough, the treatment in these texts is expeditious and superficial, directly oriented towards the obtention of Euler-Lagrange's
equations, leaving aside the question of whether the solutions are true minima or maxima, despite the importance of this distinction (for instance, while in fields such as optics one is interested in the minimal optical length, in stochastic dynamics one seeks to maximize the path entropy \cite{wang}. On the other hand, while the principle of stationary action just selects critical paths, experimentally an actual minimum is detected in some systems, see \cite{garrett}).

It is interesting to note that the principle was once called the principle of least action, although it was soon realized that many physical phenomena does not follow a trajectory that realizes a minimum of the action, but just a critical path (the main example is the harmonic oscillator, whose trajectory in phase space only minimizes the action for a time interval of length which depends on its frequency, see Sec. V of \cite{gray}, which we recommend to get details about the physical meaning of the action integral and its extremals). In view of these phenomena, and because the emphasis was on the equations of motion, the elucidation of the true nature of the critical paths of the action functional was omitted, and the interest focused on the stationary property. However, some
recent papers have made a ``call to action'' (the pun is not ours, see \cite{taylora}, \cite{moore}, \cite{gray}, \cite{gray0}), renewing the interest in their extremal properties, not only their character of paths rendering the action stationary.

Our aim in this paper is twofold: on the one hand, to offer a concise, yet rigorous and self-contained, overview of some elementary techniques of non-linear analysis to investigate the extremals of an action functional. On the other hand, we intend to show several
non trivial examples of physical interest illustrating the use of these techniques. We have avoided the well-known cases, so our examples go beyond oscillators and central potentials, and are taken from modern theories, ranging from astrophysics (Lane-Emden equations) to relativistic particles with energy dissipation.
Each one of these examples has been chosen to illustrate some particular feature. Thus, example \ref{driven} shows a Lagrangian for a dissipative system; in example \ref{n5}, as a bonus of the theory developed, we explicitly compute the solution (and its zeros) of an equation of the form $y'' +q(x)^2 y=0$ with $q(x)$ rational; example \ref{entropia}
contains a justification of the use of Lagrange multipliers in the maximum entropy principle, etc.

We also differ from previous works, such as \cite{gray} or \cite{gray0}, in the flavour of the treatment: we feel that  discussions  trying to explain some plain analytic effects in physical terms are too lengthy, and sometimes add confusion instead of enlightenment
when it comes to explicit computations. Thus, we center our exposition around the analytic definition and properties of the G\^ateaux derivatives of functionals defined by integration (Lagrange functionals) and the techniques for the study of the behaviour of solutions of differential equations such as convexity or the comparison theorems of the Sturmian theory. This will be particularly patent in section \ref{sturmian}, where we show that the main result in \cite{gray} is a direct consequence of well-known properties of the zeros of the Jacobi equation (see Proposition \ref{minimum} and comments).

We offer short proofs for those results that seem to be not common in the physics literature. The bibliography, although by no means complete, is somewhat lengthy as a result of our efforts to make it useful.

\section{Calculus of variations}
In this section we will briefly describe some basic concepts in the  
calculus of variations in 
order to set up our notation and conventions, and also in order to introduce the 
Jacobi equation and conjugate points as explicit criteria for a given extremal solution to be a minimum.  We will start by discussing G\^ateaux derivatives and extrema of functionals.
For general references on the topics of functional analysis and calculus of variations, see \cite{curtain}, \cite{flett}, \cite{gelfand}, \cite{giaquinta}, \cite{sagan}, \cite{smith}, \cite{troutman}, \cite{vanbrunt}, \cite{zeidler}. Note that we
deal with the local aspects of the theory, exclusively. There are several approaches to the global setting, some of these
were developed in \cite{shlomo}, \cite{pedroluis}, \cite{saunders}; more modern versions are developed in \cite{krupkova2}.
For detailed accounts of the theory involved in the global analysis, see \cite{krupkova1} and \cite{sardanas}.

\subsection{G\^ateaux derivatives}
Let $(E,\parallel \cdot \parallel)$ be a Banach space, $D \subseteq E$ an open subset of $E$ and $y_0 \in D$. Given a functional
$J:D \rightarrow \mathbb{R}$, if $v \in E$ is a non zero vector and $|t|$ small enough, $y_0 +tv$ will lie in $D$ so
the following definition makes sense.
\begin{defn}
Whenever it exists, the limit
$$
\delta J(y_0,v):=\lim_{t\to 0}\frac{J(y_0 +tv)-J(y_0 )}{t}
$$
is called the G\^ateaux derivative (or first variation) of $J$ at $y_0$ in the direction $v\in E$.\\
This defines a mapping $\delta J(y_0,\cdot ):E\to \mathbb{R}$. If this mapping is linear and continuous, we denote it by $J'(y_0)$ and say that $J$ is G\^ateaux differentiable at $y_0$. Thus, under these conditions, $\delta J(y_0,v)=J'(y_0)(v)$.
Another common notation is $\delta J(y_0,v)=\delta_{y_0}J(v)$.\\
The $y_0 \in D$ such that $J'(y_0 )=0$ are called critical points of the functional $J$.
\end{defn}
The extension to higher-order derivatives is immediate. If, for a fixed $v\in E$, $\delta J(z,v)$ exists for every $z\in D$, we have a mapping $D:\to \mathbb{R}$ and we can compute its G\^ateaux derivative. Given an $y_0 \in D$ and $z,v\in E$, the second
G\^ateaux derivative (or second variation) of $J$ at $y_0$ in the directions $v$ and $z$ (in that order) is
$$
\delta^2 J(y_0,v,z):=\lim_{t\to 0}\frac{\delta J_{y_0 +tz}(v)-\delta_{y_0}J(v)}{t}.
$$
If $\delta^2 J(y_0,v,z)$ exists for any $z,v\in E$, and $(v,z)\mapsto \delta^2 J(y_0,v,z)$ is bilinear and continuous, we say that
$J$ is twice G\^ateaux differentiable at $y_0 \in D$, and write $J''(y_0)$ for this mapping. With these notations, we will write
$$
\delta^2_{y_0} J(v)=\delta^2 J(y_0,v):=\delta^2 J(y_0 ,v,v).
$$
\begin{rem}\label{rem1}
For fixed $y_0 \in D$ and $v\in E$, if we consider the function $j_{(y_0,v)} :\mathbb{R}\to\mathbb{R}$ by $j_{(y_0,v)} (t)=J(y_0 +tv)$, it is obvious that it is defined in some open neighborhood of $0$, $]-\varepsilon ,\varepsilon [$, and the higher-order variations of $J$ are given by
$$
\delta^n J(y_0 ,v)=j^{(n)}_{(y_0,v)}(0).
$$
\end{rem}

We will be interested in a particular class of functionals. To introduce it, we first need a technical observation: given an open subset $U\subset\mathbb{R}^3$, the set
$$
D_U =\{ y\in\mathcal{C}^1 ([a,b]): \forall x\in [a,b],(x,y(x),y'(x))\in U\}
$$
(the prime denotes derivation, although we will also make free use of the physicist's dot notation for derivatives)
is evidently contained in the Banach space $(\mathcal{C}^1 ([a,b]),\parallel \cdot \parallel)$, endowed with the norm
$$
|| y|| = || y ||_0 +|| y' ||_0 ,
$$
where $\parallel \cdot \parallel_0$ is the supremum norm. Moreover, $D_U \subset \mathcal{C}^1 ([a,b])$ is an open subset. This follows from the fact that for a given $y_0 \in D_U$ the set $\{x,y_0(x),y'_{0}(x):x\in [a,b]\}$ is compact, so it has an open
neighborhood contained in $U$.
\begin{defn}\label{deflag}
A function $L\in\mathcal{C}^2 (U)$ is called a Lagrangian. To every Lagrangian it corresponds a functional $J:D_U \to\mathbb{R}$, called its action, defined by
$$
J(y)=\int^b_a L(x,y(x),y'(x))\mathrm{d}x.
$$
\end{defn}
\begin{prop}
For any $U\subset\mathbb{R}^3$, the action $J$ is G\^ateaux differentiable on $D_U$.
\end{prop}
\begin{proof}
Let $y\in D_U$. Taking into account the remark \ref{rem1}, note that for any $t\in]-\varepsilon ,\varepsilon [$ (applying Leibniz's theorem of derivation under the integral):
$$
j'_{(y,v)}(t)=\int^b_a \frac{\mathrm{d}}{\mathrm{d}t}\left( L(s,y(s)+tv(s),y'(s)+tv'(s)) \right) \mathrm{d}s.
$$
Evaluating the derivative at $t=0$, we get
\beq
\delta J(y,v)=\int^b_a \left( v(s)D_2 L(s,y(s),y'(s))+v'(s)D_3 L(s,y(s),y'(s)) \right) \mathrm{d}s.
\label{eqn1}
\eeq
Note that $\delta J(y,v)$ is linear in $v$. Moreover,
$$
|\delta J(y,v)|\leq ||v||\int^b_a \left( |D_2 L(s,y(s),y'(s))|+|D_3 L(s,y(s),y'(s))| \right) \mathrm{d}s,
$$
where the integral exists (because the $D_i L$, $i\in\{1,2\}$, are continuous on $U$ and $y,y'$ on $[a,b]$) and it is a number depending only on $y$,
thus constant for fixed $y$. Then, $\delta J(y,v)$ is also continuous in $v$.
\end{proof}

\subsection{Local extrema of functionals}\label{locext}
\begin{defn}
Let $J:D \rightarrow \mathbb{R}$ be a functional and let 
$y_0 \in D$. We will say that $J$ has a local maximum 
(local minimum, respectively) 
in $y_0$ if for all $y \in G$, where $G \subset D$ is a 
convex neighborhood of the point $y_0$, it follows that
\beq
\nn
J(y_0) \geq J(y) 
\eeq
($J(y_0)\leq J(y)$, respectively).
\end{defn}
\begin{thm}
Let $J:D\to\mathbb{R}$ be a functional and $y_0 \in D$. Then:
\begin{enumerate}\label{teoremaco}
\item{(Necessary condition)} If $J$ has a local extremal and the variation $\delta J(y_0 ,v)$ exists
for some $v\in E$, then $\delta J(y_0 ,v)=0$. Thus, if $J$ is G\^ateaux differentiable at $y_0$, $J'(y_0 )=0$.
\item{(Sufficient condition for a minimum)} The functional $J$ has a local minimum at $y_0$ whenever the following hold:
\begin{enumerate}[(a)]
\item\label{item1} For each $v\in E$, $\delta J(y_0 ,v)=0$.
\item\label{item2} (Coercivity) For any $y$ in a convex neighborhood of $y_0$, the second variation $\delta^2 J(y,v)$ exists for each
$v\in E$. Moreover, there exists a $c>0$ such that
$$
\delta^2 J(y_0 ,v)\geq c||v||^2,
$$
for all $v\in E$.
\item{(Weak continuity)}\label{item3} Given $\varepsilon >0$, there exists an $\eta >0$ such that
$$
|\delta^2 J(y ,v)-\delta^2 J(y_0 ,v)|\leq \varepsilon ||v||^2 ,
$$
for any $v\in E$ and $y$ satisfying $||y-y_0 ||<\eta$.
\end{enumerate}  
\end{enumerate}
\end{thm}
\begin{proof}\leavevmode
\begin{enumerate}
\item Consider $j_{(y_0,v)} (t)=J(y_0 +tv)$, so if $J$ has a local extremal at $y_0$, $j_{(y_0,v)}$ has a local extremum at
$t=0$. Then, it must be (recall remark \ref{rem1}) $0=j'_{(y_0,v)} (0)=\delta J(y_0 ,v)$.
\item Suppose each of \eqref{item1},\eqref{item2},\eqref{item3} holds. As before, we have $j'_{(y_0,v)} (t)=\delta J(y_0 +tv,v)$ and $j''_{(y_0,v)} (t)=\delta^2 J(y_0 +tv,v)$. The hypothesis on the second derivatives of $J$ allows us to develop
$j_{(y_0,v)} (t)$ by Taylor in the interval $[0,1]$, and there exists a $\xi\in ]0,1[$ such that
$$
J(y_0 +v)-J(y_0)=j_{(y_0,v)} (1)-j_{(y_0,v)} (0)=\frac{1}{2}j''_{(y_0,v)} (\xi).
$$
As $j''_{(y_0,v)} (\xi)=\delta^2 J(y_0 +\xi v,v)$, we have the following bound:
\begin{align*}
& J(y_0 +v)-J(y_0)= \\
& \frac{1}{2}\delta^2 J(y_0,v)+\frac{1}{2}\left( \delta^2 J(y_0 +\xi v,v)- \frac{1}{2}\delta^2 J(y_0,v)\right)\geq \\
& \frac{1}{2}c||v||^2 + \frac{1}{2}\left( \delta^2 J(y_0 +\xi v,v)- \frac{1}{2}\delta^2 J(y_0,v)\right)
\end{align*}
Taking $\varepsilon =c/4$ in \eqref{item3}, there exists an $\eta >0$ such that
$$
|\delta^2 J(y ,v)-\delta^2 J(y_0 ,v)|\leq \frac{c}{4} ||v||^2 ,
$$
for $||y-y_0 ||<\eta$. Choosing now $||v||<c\eta /2$, it is $||y_0 +\xi v-y_0 ||\leq ||v||<c\eta /2$, so
$$
|\delta^2 J(y_0 +\xi v ,v)-\delta^2 J(y_0 ,v)|\leq \frac{c}{4} ||v||^2 .
$$
Substituting:
$$
J(y_0 +v)-J(y_0)\geq \frac{1}{2}c||v||^2 -\frac{1}{4}c||v||^2=\frac{1}{4}c||v||^2 >0,
$$
so $J$ has a local minimum.
\end{enumerate}
\end{proof}
\begin{rem}
Writing $c<0$ and reversing the inequality for $\delta^2 J(y_0 ,v)$ in \eqref{item2}, we get sufficient conditions for a local maximum.
\end{rem}
\begin{rem}\label{localmin}
Note that condition \eqref{item1} alone does not guarantee the existence of a local minimum, see counter-examples in
\cite{pars}, \S 2.10.
\end{rem}
It is interesting to particularize the condition $\delta J(y_0 ,v)=0$ to the case of an action functional. For this, we need a couple of technical results whose proof is straightforward, but anyway can be found in any of the references cited at the beginning of this section. We will denote
$\mathcal{C}^k_0 ([a,b])=\{ f\in \mathcal{C}^k ([a,b]): f(a)=0=f(b)\}$.
\begin{lem}[Lagrange]\label{lema1}
Let $f\in\mathcal{C}([a,b])$ be a continuous real-valued 
function over the interval $[a,b]$ such that
\beq
\nn
\int_{a}^{b}f(x)\mu(x)\, \mathrm{d}x=0
\eeq
for all $\mu\in\mathcal{C}_{0}([a,b])$. Then it follows $f\equiv 0$.
\end{lem}

\begin{lem}[DuBois-Reymond]\label{lema2} Let $f\in\mathcal{C}([a,b])$ and $g\in\mathcal{C}^{1}([a,b])$ such
that
\[
\int_{a}^{b}(f(x)\mu(x)+g(x)\mu'(x))\, \mathrm{d}x=0
\]
for all $\mu\in\mathcal{C}_{0}^{1}([a,b])$. Then it follows $g'=-f$.
\end{lem}
Now, a simple integration by parts in \eqref{eqn1}, and the application of Lemmas \ref{lema1}, \ref{lema2} and Theorem \ref{teoremaco},
leads directly to the following result.
\begin{thm}[Euler-Lagrange]\label{eulero}
If $y\in D$ is an extremal (maximum or minimum) for the action functional $J:D\rightarrow\mathbb{R}$ given as 
in Definition \ref{deflag}, then $y$ must 
satisfy the Euler-Lagrange equations
\beq
D_2 L(x,y(x),y'(x))-\frac{\mathrm{d}}{\mathrm{d}x} D_3 L(x,y(x),y'(x))=0 \,.
\label{eq:EulerLagrange}
\eeq
\end{thm}
\begin{rem}
In physics literature, it is common to commit a slight abuse of notation and to write the Euler-Lagrange equations
in the form
$$
\dfrac{\partial }{\partial y}L(x,y(x),y'(x))-\frac{\mathrm{d}}{\mathrm{d}x}\dfrac{\partial}{\partial y'}L(x,y(x),y'(x))=0 \,.
$$
\end{rem}

Note that, for the case at hand, writing $L(s,y(s)+tv(s),y'(s)+tv'(s))=L(s)$ for simplicity:
$$
j''_{(y,v)} (t)=\int^b_a (v^2 (s)D_{22}L(s)+2v(s)v'(s)D_{23}L(s)+(v')^2(s)D_{33}L(s))\mathrm{d}s,
$$
so, evaluating at $t=0$,
\beq
&\delta^2 J(y,v)= \int^b_a (v^2 (s)D_{22}L(s,y(s),y'(s)) \label{variacion2} \\
&+2v(s)v'(s)D_{23}L(s,y(s),y'(s))+(v')^2(s)D_{33}L(s,y(s),y'(s)))\mathrm{d}s .\notag
\eeq
It is now a routine computation (continuity arguments and Schwarz inequality) to prove that for an action
functional $J(y)=\int^b_a L(x,y,y')\mathrm{d}x$ with $L\in \mathcal{C}^2 (U)$ such that its second partial derivatives are bounded on $U$, under the hypothesis \eqref{item1} and \eqref{item2} of Theorem \ref{teoremaco}, the condition \eqref{item3} is satisfied (\cite{giaquinta}, pg. 224). Thus, a path $y_0 \in D$ is a minimum if it satisfies Euler-Lagrange's equations \eqref{eq:EulerLagrange} \emph{and} the second G\^ateaux differential at $y_0$ is coercive, that is, there exists a $c>0$ such that, for all $v\in E$:
\beq
\delta^2 J(y_0 ,v)\geq c||v||^2 .
\label{coercive}
\eeq

\subsection{Problems with constraints}\label{sec:constraints}
The calculus of variation is frequently applied when there are constraints. The problem can be reduced to
that of extremizing a single functional constructed out from the original one and the constraints.
\begin{defn}
Let $J$ be a functional defined on a Banach space $(X,\parallel \cdot \parallel )$.
We say that $\delta J$ is weakly continuous at $y\in X$ if:
\begin{enumerate}[(a)]
\item The domain of $J$ contains an open neighborhood $D\ni y$, and, for each $h\in X$, the variation $\delta J(y,h)$ is defined.
\item $
\lim_{z\to y}\delta J(z,h)=\delta J(y,h).
$
\end{enumerate}
If there exists an $r>0$ such that $\delta J$ is weakly continuous for every $z\in B(y;r)$, we say that $\delta J$
is locally weakly continuous at $y$, or simply weakly continuous near $y$.
\end{defn}

When we have an open subset $U=]a,b[\times \mathbb{R}\times\mathbb{R}\subset\mathbb{R}^3$ and an
action functional $J:D_U \to \mathbb{R}$, $J(y)=\int^b_a L(x,y,y')\mathrm{d}x$, it is easy to see that
imposing some mild conditions on the Lagrangian $L\in\mathcal{C}^2 (U)$ we obtain a weakly continuous
functional. For instance, it is enough to require that the second partial derivatives of $L$ be bounded on $U$,
or that its first partial derivatives be uniformly continuous. For most of the actions appearing in physics,
however, it is usually easier to prove the weak continuity directly from the definition (cfr. Example \ref{entropia}).

Let now $J=:K_0,K_1,...,K_r$ be functionals defined on $D$, all of them G\^ateaux differentiables at each point $y\in D$.
We will assume that the set
$$
S=\{ y\in D: K_i (y)=k_i ,1\leq i\leq r \}
$$
is not void, and that $y_0$ is an interior point of $S$ such that $J|_S$ has a local extremal at $y_0$.
\begin{prop}\label{multipliers}
Let $\delta K_j$ be weakly continuous near $y_0$, $0\leq j\leq r$. Then, for any $h=:h_0,h_1,...,h_r \in X$ we have:
$$
\det (\delta K_j(y_0,h_m))=0,\quad 0\leq j,m\leq r.
$$
\end{prop}

\begin{proof}
By reduction to the absurd. Let us assume that there exist $h,h_1,...,h_r$ such that the determinant is non zero.
As $y_0$ has an open neighborhood $D$, there exist a set of scalars $\alpha ,\beta_1 ,...,\beta_r$ such that
$y_0 +\alpha h+\beta_1 h_1 +\cdots \beta_r h_r \in D$, and the variations $\delta J,\delta K_1,...,\delta K_r$ are continuous at $y_0 +\alpha h+\beta_1 h_1 +\cdots \beta_r h_r$. Let us define the function $F:\mathbb{R}^{r+1}\to \mathbb{R}^{r+1}$, on a neighborhood $G$ of the origin $(0,...,0)\in \mathbb{R}^{r+1}$, by
$$
F_{p+1} (\alpha ,\beta_1 ,...,\beta_r )=K_p (y_0 +\alpha h+\beta_1 h_1 +\cdots \beta_r h_r),
$$
for $0\leq p\leq r$ (remember $K_0 =J$ and $h_0 =h$). It is immediate that
$$
D_{q+1} F_{p+1} (\alpha ,\beta_1 ,...,\beta_r )=\delta K_p (y_0 +\alpha h+\beta_1 h_1 +\cdots \beta_r h_r ,h_q),
$$
for $0\leq q\leq r$. Now, as $\delta J,\delta K_1,...,\delta K_r$ are continuous near $y_0$, by shrinking $G$ if necessary we
can assume $F\in \mathcal{C}^1 (G)$, with Jacobian at the origin:
$$
\det (\delta K_j(y_0,h_m))\neq 0,\quad 0\leq j,m\leq r
$$
by hypothesis. Applying to $F$ the inverse function theorem in the neighborhood of the origin, we get that there exists
an open subset $V\subset \mathbb{R}^{r+1}$, containing $F(0,...,0)=(J(y_0),k_1 ,...,k_r )$, and a local diffeomorphism
$\varphi :V\to G$ such that $\tilde{G}=\varphi (V)\subset G$ is an open neighborhood of the origin in $\mathbb{R}^{r+1}$,
and for all $(x,y_1,...,y_r)\in V$:
$$
(x,y_1,...,y_r)=F(\varphi (x,y_1,...,y_r)).
$$
In particular, $\varphi (J(y_0),k_1,...,k_r)=(0,0,...,0)$. As $(J(y_0),k_1,...,k_r)$ is a point of the open set $V\subset \mathbb{R}^{r+1}$, we can find in $V$ two different points $(x_1,k_1,...,k_r)$ and $(x_2,k_1,...,k_r)$ such that $x_1 <J(y_0)<x_2$. Their corresponding images by $\varphi$ are $(\alpha^1 ,\beta^1_1,...,\beta^1_r)=\varphi (x_1,k_1,...,k_r)$ and
$(\alpha^2 ,\beta^2_1,...,\beta^2_r)=\varphi (x_2,k_1,...,k_r)$. Thus, the vectors $u=y_0 +\alpha^1 h+\beta^1_1 h_1+\cdots +\beta^1_r h_r$ and $v=y_0 +\alpha^2 h+\beta^2_1 h_1+\cdots +\beta^2_r h_r$ belong to $D$. Moreover,
\begin{align*}
&F(\alpha^1 ,\beta^1_1,...,\beta^1_r)=(x_1,k_1,...,k_r), \\
&F(\alpha^2 ,\beta^2_1,...,\beta^2_r)=(x_2,k_1,...,k_r),
\end{align*}
so, equating components:
\begin{align*}
&J(u)=K_0 (u)= F_0 (\alpha^1 ,\beta^1_1,...,\beta^1_r)=x_1,\\
&K_p (u)=F_p (u)=F_p (\alpha^1 ,\beta^1_1,...,\beta^1_r)=k_p,
\end{align*}
for $1\leq p\leq r$, and
\begin{align*}
&J(v)=K_0 (v)= F_0 (\alpha^2 ,\beta^2_1,...,\beta^2_r)=x_2,\\
&K_p (v)=F_p (v)=F_p (\alpha^2 ,\beta^2_1,...,\beta^2_r)=k_p,
\end{align*}
so $u,v\in S$ too. But, because of our choices, $J(u)=x_1<J(y_0)<x_2=J(v)$, and this construction can be repeated for
$x_1,x_2$ arbitrarily close to $J(y_0)$, so (by the continuity of $\varphi$) the corresponding vectors $u,v$ will be
arbitrarily close to $y_0$, contradicting the assumption that $y_0$ is a local extremal of $J|_S$.
\end{proof}

\begin{thm}\label{th:multipliers}
Let $\delta K_j$ be weakly continuous near $y_0$. Then, either
$$
\det (\delta K_i(y_0,h_l))=0,\quad 1\leq i,l\leq r,
$$
or there exist a set of real numbers (the Lagrange multipliers) $\lambda_1,...,\lambda_r$ such that
$$
\delta J(y_0 ,h)=\sum^m_{i=1}\lambda_i \delta K_i (y_0,h).
$$
\end{thm}
\begin{proof}
By Proposition \ref{multipliers}, $\det (\delta K_j(y_0,h_j))=0,\quad 0\leq j\leq r$; developing by the first column (where $K_0 =J$ and $h_0 =h$), we get
\beq\label{aux1}
\delta J(y_0 ,h)\det (\delta K_i (y_0,h_l))+\sum^m_{i=1}\mu_i \delta K_i (y_0,h)=0,
\eeq
for some set of scalars $\mu_1,...,\mu_r$ (which depend on the vectors $(h_1,...,h_r)$).\\
Then, either 
$$
\det (\delta K_i(y_0,h_l))=0,\quad 1\leq i,l\leq r,
$$
or there exist some set of vectors $v_1,...,v_r\in X$ such that $\det (\delta K_i(y_0,v_l))\neq 0$ In this case,
substituting in \eqref{aux1}, we obtain
$$
\delta J(y_0,h)=\sum^m_{i=1}\lambda_i \delta K_i (y_0 ,h),
$$
where
$$
\lambda_i =-\frac{\mu_i}{\det (\delta K_i (y_0,h_l))}.
$$
\end{proof}

\section{Conjugate points}
Let us rewrite the conditions for a minimum of the action $J$, found in the previous section,
in terms of a differential equation involving the derivatives of the Lagrangian $L$.
\begin{prop}
Let the action $J$ be given as in Definition~\ref{deflag},
and let $y_0 \in D$.  
Then, the second variation of $J$ at $y_0$, in the direction of a $v\in\mathcal{C}^1_0 ([a,b])$,
$\delta^2 J(y_0,v)$, reads
\beq
\delta^2 J(y_0,v)=\frac{1}{2}\int^b_a \left( Pv'^{\, 2} +Qv^2 \right) \, \mathrm{d}x \,,
\label{eq:2ndvariation}
\eeq
where the functions $P(x)$ and $Q(x)$ are explicitly given by
\beq
P(x) &=& D_{33}L(x,y(x),y'(x)) \label{eq:Pfunc} \\
Q(x) &=& D_{22}L(x,y(x),y'(x))-\frac{\mathrm{d}}{\mathrm{d}x}D_{23}L(x,y(x),y'(x))\label{eq:Qfunc}
\eeq
\end{prop}
\begin{proof}
Just make an integration by parts in the middle term of the integrand in \eqref{variacion2}, taking
into account the boundary conditions on $v\in \mathcal{C}^1_0 ([a,b])$.
\end{proof}
\begin{rem}
The notation used in physics is:
\begin{align*}
P &= \frac{\partial^2 L}{\partial (y')^2} \,,\\
Q &= \frac{\partial^2 L}{\partial y^2}-\frac{\mathrm{d}}{\mathrm{d}x}\left( \frac{\partial^2 L}{\partial y\partial y'} \right) \,.
\end{align*}
\end{rem}

Lagrange considered equation \eqref{eq:2ndvariation} already in $1786$. He thought that a sufficient condition to have a minimum
would be the positivity of the second variation $\delta^2 J(y_0,v)>0$ (which is \emph{not} true: coercivity is needed), so
he tried to ``complete the square'' in~\eqref{eq:2ndvariation} by introducing a 
boundary term of the form $gv^2/2$, where 
$g\in\mathcal{C}([a,b])$ is to be determined.  
In this way, we have
\beq
\hspace{-4ex}
\delta^2 J(y_0 ,v) &=& \int^b_a \left( P(v')^2 +Qv^2 \right)\, \mathrm{d}x+\int^b_a \frac{\mathrm{d}}{\mathrm{d}x}(gv^2 )\, dx  \nn\\
&=& \int^b_ a \left( P(v')^2 +2gvv'+(g'+Q)v^2 \right) \, \mathrm{d}x \\
&=& \int^b_ a P\left( v'+\frac{g}{P}v \right)^2 \, \mathrm{d}x+\int^b_a \left(g'+Q-\frac{g^2}{P} \right)v^2 \, \mathrm{d}x   \,.\nn
\label{desarrollolag}
\eeq
Thus, it is straightforward that $\delta^2 J(y_0,v)$ 
will be positive definite if the following conditions are satisfied
\beq
&&P=D_{33}L(x,y(x),y'(x))> 0 \,,
\label{CondLegendre1}      \\
&&g'+Q-\frac{g^2}{P}=0 \mbox{ has a solution }g \,.
\label{CondLegendre2}  
\eeq
Condition~\eqref{CondLegendre1} is known as 
the Legendre condition.  Also, note that equation~\eqref{CondLegendre2} is
of Riccati type. This equation is basic to determine the extremality properties of critical points of the action.

\begin{defn}
Let $J$ be an action functional, 
and let $f\in\mathcal{C}([a,b])$. The differential equation for $f$
\beq
\label{eq:Jacobi}
-\frac{\mathrm{d}}{\mathrm{d}x}\left( P\frac{\mathrm{d}f}{\mathrm{d}x} \right) +Qf=0 \,,
\eeq
where $P$ and $Q$ are given in~\eqref{eq:Pfunc} and~\eqref{eq:Qfunc},
respectively, is called the Jacobi equation.
\end{defn}
Notice that Jacobi equation is simply obtained by introducing
the change of variable $g=-P\mathrm{d}(\ln f)/\mathrm{d}x$ in 
equation~\eqref{CondLegendre2}, which renders it linear. Once we solve the equation for $f$, we get $g$ and then we can assure that, if $P>0$, then $\delta^2 J(y_0 ,v)>0$. Although we know that this is not enough to guarantee a minimum
(recall Remark \ref{localmin}), the properties of
the solutions to the Jacobi equation \eqref{eq:Jacobi} will lead us to an equivalent condition for a minimum
of the action, given in terms of quantities determined by the Lagrangian.

\begin{defn}
Two points $p,q \in \mathbb{R}$ (with $p<q$) are called 
conjugate with  
respect to the Jacobi equation~\eqref{eq:Jacobi} 
if there is a solution $f\in\mathcal{C}^2([a,b])$ of
\eqref{eq:Jacobi} such that $\left. f\right|_{]p,q[}\neq 0$ and $f(p)=0=f(q)$.
\end{defn}

The following result is just a particular case of K. Friedrichs' inequalities for the one-dimensional case
(see \cite{adams}). Its proof can also be done directly, as an application of the Schwarz inequality.
\begin{lem}\label{lema4}
For any $v\in\mathcal{C}^1_0 ([a,b])$, we have
$$
\int^b_a v^2 (x)\mathrm{d}x \leq \frac{(b-a)^2}{2}\int^b_a (v')^2 (x)\mathrm{d}x 
$$
\end{lem}

\begin{thm}
Let $J(y)$ be an action functional as in Definition~\ref{deflag}.
Sufficient conditions for a critical point $y_0$ of 
$J(y)$ to be a local minimum in the interval 
$[a,b]$ 
are given by
\begin{enumerate}[(a)]
\item For all $x \in [a,b]$
\beq
\nn
P(x)=D_{33}L(x,y(x),y'(x))>0  \,.
\eeq
\item The interval $[a,b]$ does not contain conjugate points 
at $x=a$ with respect to Jacobi equation~\eqref{eq:Jacobi}.
\end{enumerate}
\end{thm}

\begin{proof}
Recall from \eqref{desarrollolag} the expression
$$
\delta^2 J(y_0 ,v) = \int^b_ a \left( P(v')^2 +2gvv'+(g'+Q)v^2 \right) .
$$
Because of the assumptions made on the continuity of the derivatives of $L$ and the compactness of $[a,b]$,
we can choose a number $\sigma$ such that $0\leq \sigma <\min_{[a,b]}\{ P(x)\}$. Inserting $\sigma P(v')^2-\sigma P(v')^2$
in the equation above, and repeating the computation in \eqref{desarrollolag} gives
\beq \label{sdelta}
\delta^2 J(y_0 ,v) &=& \int^b_ a  (P-\sigma )\left( v'+\frac{g}{P-\sigma}\right)^2\mathrm{d}x \\
&+&\int^b_a \left( g' +Q-\frac{g^2}{P-\sigma}\right) v^2\mathrm{d}x
+\sigma\int^b_a (v')^2 \mathrm{d}x .\nn
\eeq
As $P(x)>0$ and we have chosen $\sigma$ such that $P(x)-\sigma >0$ on $[a,b]$, the first integral is positive, as it is the third one. In order to cancel out the second integral, we must take a $g\in\mathcal{C}^1 ([a,b])$ such that
$$
g'+Q-\frac{g^2}{P-\sigma}=0.
$$
Introducing a function $f\in\mathcal{C}^2 ([a,b])$ through
\beq
g=-\frac{f'}{f}(P-\sigma ),
\label{auxiliar1}
\eeq
we arrive at the equation for $f$:
\beq
-\frac{\mathrm{d}}{\mathrm{d}x}\left( (P-\sigma )\frac{\mathrm{d}f}{\mathrm{d}x} \right) +Qf=0.
\label{auxiliar2}
\eeq
By the theorem on the dependence on parameters of the solutions to a second order differential equation, the
general solution of \eqref{auxiliar2} can be written as $f(x,\sigma )$, with $f(x,0)=f(x)$. Note that, by hypothesis,
$f(x,0)$ does not admit points conjugate to $a$ in $[a,b]$, so (by continuity), neither does $f(x,\sigma )$ for $\sigma >0$
but close enough to $0$. If $\tilde{f}(x)=f(x,\sigma )$ is such a solution, by substituting the corresponding $\tilde{g}$ of
equation \eqref{auxiliar1} into \eqref{sdelta}, we get:
\begin{align*}
\delta^2 J(y_0 ,v) &= \int^b_ a  (P-\sigma )\left( v'+\frac{\tilde{g}}{P-\sigma}\right)^2\mathrm{d}x 
+\sigma\int^b_a (v')^2 \mathrm{d}x \\
&\geq \sigma\int^b_a (v')^2 \mathrm{d}x.
\end{align*}
Now, applying Lemma \ref{lema4},
$$
\delta^2 J(y_0 ,v)\geq \frac{\sigma}{1+\frac{(b-a)^2}{2}}||v||^2:=c||v||^2 .
$$
The statement follows then from Theorem \ref{teoremaco} (see also the comments at the end of subsection \ref{locext}).
\end{proof}
\begin{rem}
We can obtain a criterion for a local maximum just by considering the condition $P(x)=D_{33}L(x,y(x),y'(x))<0$ and repeating the computations in the theorem with the inequalities reversed.
\end{rem}
\begin{rem}\label{supermethod}
In order to apply the criterion of conjugate points in practice, it is desirable to have at our disposal
some tools for explicitly computing solutions of the Jacobi equation \eqref{eq:Jacobi}. An old (but useful) method
(\cite{pars} pp. 56--57, \cite{fox} pp.42--43) is the following: the general solution of Euler-Lagrange's equations
(which are second order) has the form $y=y(x;\alpha ,\beta )$, where $\alpha ,\beta$ are constants of integration
(on which the $y$ dependence is differentiable, under some mild conditions. In the examples this will be obvious).
Then, $D_2 y(x;\alpha ,\beta )\equiv \frac{\partial y}{\partial \alpha}$ and $D_3 y(x;\alpha ,\beta )\equiv \frac{\partial y}{\partial \beta}$ are two independent solutions of the Jacobi equation\footnote{The
proof is extremely simple: just take derivatives with respect to, say, $\alpha$ in Euler-Lagrange's equations
(applying the chain rule) and collect terms, taking into account that
$$
\frac{\mathrm{d}}{\mathrm{d}x}\left( \frac{\partial^2 L}{\partial y\partial y'}\frac{\partial y}{\partial \alpha} \right)
=\frac{\mathrm{d}}{\mathrm{d}x}\left( \frac{\partial^2 L}{\partial y\partial y'}\right) \frac{\partial y}{\partial \alpha}
+ \frac{\partial^2 L}{\partial y\partial y'}\frac{\partial y'}{\partial \alpha}.
$$}. We will use this method in example \ref{n5}.

\end{rem}

\section{Convex functionals}

In this section we will discuss the particular case of convex Lagrangians.  
We will start by recalling that a subset $S \subset \mathbb{R}^n$ is said to be convex if 
for all $p,q \in S$ the interval $[p,q]$ lies entirely 
inside of $S$.  This is equivalent to say that
\begin{equation}
[p,q]:= \lbrace p + t(q-p) = tq + (1-t)p : 0 \leq t \leq 1 \rbrace \subset S.
\end{equation}

\begin{defn}
A function $f:S \subset \mathbb{R}^2$ defined over a convex 
set, is said to be convex if
$$
f([p,q]) \leq [f(p),f(q)]
$$ 
or, equivalently, for all $t \in [0,1]$: 
\beq
f(p + t(q-p)) \leq f(p)+t(f(q)-f(p))
\label{convex1}
\eeq
\end{defn}

Let $f:S\to\mathbb{R}$ be a convex \emph{differentiable} function. For any $t\in [0,1]$, $p,q\in S$ we have
\eqref{convex1}, and, on the other hand, by the intermediate value theorem, there exists a $w_t \in [p,p+t(q-p)]$ such that
\beq
f(p+t(q-p))=f(p)+t\mathrm{d}_{w_t}f(q-p),
\label{convex2}
\eeq
where $\mathrm{d}_{w_t}f$ denotes the differential of $f$ at $w_t$.
From \eqref{convex1} and \eqref{convex2} we get $f(p)+t\mathrm{d}_{w_t}f(q-p) \leq tf(q)+(1-t)f(p)$, that is:
$$
f(p)+\mathrm{d}_{w_t}f(q-p)\leq f(q).
$$
Taking the limit $t\to 0$ implies $w_t \to p$, so:
\beq
f(p)+\mathrm{d}_{p}f(q-p)\leq f(q).
\label{convex3}
\eeq

\begin{rem}
There exist well-known criteria (in terms of the Hessian matrix) to decide whether a function
$f:S\to \mathbb{R}$ is convex or not, see for instance \cite{vanbrunt} sec. 10.7. We will apply
one such criterion in Example \ref{lane0}.
\end{rem}

As a straightforward application of these results, we have the following theorem, stating that the critical points
of an action with a convex Lagrangian are always minimals.

\begin{thm}
\label{thm:convex}
Consider a set $U=[a,b]\times S\subset \mathbb{R}^3$, such that for each fixed $x\in [a,b]$ the set
$S_x =\{ (x,u,v)\in U\}\subset \mathbb{R}^2$ is convex. Suppose that for any $x \in [a,b]$, the Lagrangian function
$L(x,\cdot ,\cdot ):S_x \to\mathbb{R}$ is convex. Then, any critical path $y_0=y_0(x)$ is a minimal solution, among the 
paths with the same endpoints, for the corresponding action functional $J(y)=\int^b_a L(x,y(x),y'(x))\mathrm{d}x$.
\end{thm}

\begin{proof}
The hypothesis of convexity implies, by \eqref{convex3},
\beq
&&L(x,u_2,v_2) \geq L(x,u_1,v_1)+\mathrm{d}_{(x,u_1,v_1)}L((u_2,v_2)-(u_1,v_1)) \label{convex4}\\
&&= L(x,u_1,v_1)+D_1 L(x,u_1,v_1)(u_2 -u_1)+ D_2  L(x,u_1,v_1)(v_2 -v_1).\nn 
\eeq
Now we compute the action on a critical path, $y_0 (x)$, and an arbitrary nearby one $y (x)$, with the 
same endpoints ($y_0 (a)=y(a)$ and $y_0(b)=y(b)$), and compare them. From \eqref{convex4}:
\beq
&&J(y)-J(y_0)=\int^b_a (L(x,y(x),y'(x))-L(x,y_0 (x),y'_0 (x)))\mathrm{d}x \nn \\
&&\geq \int^b_a (D_1L(x,y_0,y'_0)(y -y_0)+D_2L(x,y_0,y'_0)(y'-y'_0))\mathrm{d}x.
\eeq
The second term in the last integrand can be --as usual-- integrated by parts, we then get:
$$
J(y)-J(y_0)\geq \int^b_a \left( D_1L(x,y_0,y'_0)-\frac{\mathrm{d}}{\mathrm{d}x}D_2L(x,y_0,y'_0) \right) (y -y_0)\mathrm{d}x .
$$
But, by hypothesis, $y_0 (x)$ is a critical path; equivalently, for each $x\in [a,b]$ it satisfies the Euler-Lagrange equations \eqref{eq:EulerLagrange}, and this implies $J(y)\geq J(y_0)$.
\end{proof}
\begin{rem}\label{concave}
Reversing the inequalities, we obtain the corresponding result for concave functionals.
\end{rem}

\section{Sturmian Theory}\label{sturmian}
\label{sec:sturm}
Sturmian theory is concerned with the analysis of the zeros that
a solution of a linear second order differential equation, of the 
form
\beq
\label{eq:sturm}
\dfrac{d^2y}{dx^2}+p(x) \dfrac{dy}{\mathrm{d}x} + q(x) y =0,
\eeq
(with $p(x),q(x)$ piece-wise continuous) has in a given interval of the independent variable. This theory is
an invaluable tool to check the properties of critical points, as we will see in the examples of the next section.\\
The results stated here without proof are well-known (see \cite{agarwal}, \cite{simmons}). We write them just for easy reference.\\
As it is well-known, the differential equation~\eqref{eq:sturm} 
may be written, through the change of variable $v=y\exp (\frac{1}{2}\int p\mathrm{d}x)$, in its normal form
\beq
\label{eq:normal}
\dfrac{\mathrm{d}^2v}{\mathrm{d}x^2} + r(x) v =0  \,,
\eeq 
where $r(x) = q(x) - \frac{1}{4} p^2(x) - \frac{1}{2} p'(x)$, which 
clearly preserves the zeros of the solutions to~\eqref{eq:sturm} if
$\int^x p(s)\mathrm{d}s$ is finite for finite $x$. The first observation is that the zeros of such an
equation cannot accumulate.

\begin{prop}\label{minimum}
Let $y(x)$ be a non trivial solution of \eqref{eq:sturm} or \eqref{eq:normal}. Then, its zeros are simple and the set they form does not have
accumulation points. Thus, on each closed interval $[a,b]$, $y(x)$ only possess a finite number of zeros. 
\end{prop}
\begin{proof}
If $x_0$ were a double zero, $y(x_0)=0=y'(x_0)$, so $y(x)$ would be the trivial solution by uniqueness.\\
If $x_0$ were an accumulation point for the zeros of $y(x)$, there would be a sequence of zeros $(x_n)$ such that
$x_n \to x_0$. By Rolle's theorem, there is a zero of the derivative between any two consecutive zeros of the function, so there would be a sequence $(u_m)$ of zeros of $y'(x)$ such that $u_m \to x_0$. Then, by continuity of both $y$ and $y'$, we would have $y(x_0)=0=y'(x_0)$, which, as we have just seen, is impossible.
\end{proof}
\noindent As a corollary, the zeros of a non trivial solution of \eqref{eq:sturm} or \eqref{eq:normal} are: either a finite set, or a
sequence diverging to $+\infty$, or a sequence diverging to $-\infty$, or a sequence diverging to $\pm \infty$.\\
This applies in particular to the Jacobi equation \eqref{eq:Jacobi}. Thus, if we have a solution $f(x)$ on the interval $[a,b]$ such that $f(a)=0$, its first zero after $x=a$ must be located at a point $c>a$. In other words: there exists a
$c>a$ such that there are no conjugate points in the interval $[a,c]$.
A direct consequence of this fact is that for a short enough interval, critical points of the action 
$J(y)=\int^b_a L(x,y,y')\mathrm{d}x$ \emph{such that $P(x)=D_{33} L(x,y,y')>0$}, are local minimizers.\\
Some authors state, erroneously, that \emph{for any} Lagrangian the critical points are local minimizers. The origin of the confusion can be traced back to the fact that this is true for \emph{natural} Lagrangians\footnote{That is, those of the form $L(y,y')=K(y')-V(y)$ where $K$ is a positive-definite quadratic form associated to some metric (usually $K(y')=(y')^2/2$, that is, the metric is the euclidean one) and $V$ is some $\mathcal{C}^1$ function. Note that in this case $P\geq 0$.}, in particular it is true for free Lagrangians ($V=0$), for which the trajectories are geodesics. However, not every system of interest in Physics is natural. Of all the examples presented in this note, only Example \ref{qpotential} is natural; and in Example \ref{entropia} we present a case for which the critical points are maximizers.


\begin{defn}
If every solution $y$ of \eqref{eq:sturm} or \eqref{eq:normal} has
arbitrarily large (in absolute value) zeros, then the
equation (and all its solutions) are called oscillatory. Otherwise, the equation
and all of its solutions are called non-oscillatory.
\end{defn}

\begin{thm}
Let be $y_1(x)$ and $y_2(x)$ denote two linear independent solutions of
equation~\eqref{eq:normal}.  Then, the zeros of both functions are distinct and alternating in the sense that $y_1(x)$ has a 
zero between any two consecutive zeros of $y_2(x)$, and vice-versa.
\end{thm}

\begin{thm}\label{nozeros}
If $r(x) \leq 0$ on $[a,b]$, then no non-trivial solution of \eqref{eq:normal} can have
two zeros on $[a,b]$.
\end{thm}

\begin{thm}[Sturm's Comparison Theorem]\label{th:sturm}
Let $y_1(x)$ and $y_2(x)$ be non-trivial solutions to the 
differential equation~\eqref{eq:normal} with 
$r_1(x)$ and $r_2(x)$, respectively. 
If $r_1>r_2$ in a certain interval $[a,b]$,
then $y_1(x)$ has at least a zero between two consecutive 
zeros of $y_2(x)$, unless $y_1 =y_2$ on $[a,b]$.
\end{thm}

\section{Examples}
In this section we will develop some physically motivated examples
in order to elucidate the ideas analysed so far. In each case, the regularity conditions on the Lagrangian are trivially satisfied. Unless otherwise explicitly stated, we will work on the space $X=\mathcal{C}^1([a,b])$. Also, in some examples
we will follow the notation common in physics, taking $t$ as the independent variable and $x$, $\dot{x}=\mathrm{d}x/\mathrm{d}t$
as the dependent ones.

\subsection{Driven harmonic oscillator}\label{driven}
As it is well known, the differential equation for a 
driven damped harmonic oscillator under a sinusoidal 
external force is given by~\cite{goldstein}, \cite{hand}, \cite{saletan},
\cite{taylor}, \cite{thornton}:
\beq\label{eq:driven}
\ddot{x} + \beta \dot{x} + \omega^2_0 x = \sin( \omega t)
  \,.
\eeq
Its solutions are well-known. They have the form
\beq\label{soldamped}
x(t)=\frac{1}{\omega Z}\sin (\omega t+\varphi )\, ,
\eeq
where $Z=\sqrt{\beta^2 +\frac{1}{\omega^2}(\omega^2_0 -\omega^2)}$ is the impedance and
$\varphi =\arctan \left( \frac{\beta \omega}{\omega^2_0 -\omega^2} \right)$ is the phase.\\
The corresponding Lagrangian function is:
$$
L(t,x,\dot{x}) = \frac{1}{2} e^{\beta t} \left( \dot{x}^2 - 
\frac{\beta \sin(\omega t) -
\omega \cos(\omega t)}{\omega^2 + \beta^2}\dot{x} - \omega^2_0 x^2 \right) \,.
$$
Here we follow the standard conventions and denote 
$\omega_0\in\mathbb{R}$ as the natural oscillation frequency, $\beta\in\mathbb{R}^+$ is the damping parameter, and
$\omega\in\mathbb{R}$ stands
for the frequency of the driving force. For this Lagrangian,
we straightforwardly note that the functions $P=e^{\beta t}>0$ and 
$Q=- \omega^2_0 e^{\beta t}$, given by~\eqref{eq:Pfunc} 
and~\eqref{eq:Qfunc}, respectively, lead to the Jacobi equation
\beq
\dfrac{\mathrm{d}^2f}{\mathrm{d}t^2} + \beta \dfrac{\mathrm{d}f}{\mathrm{d}t} + \omega^2_0f = 0
\eeq
which is the damped harmonic oscillation equation for 
the function $f$.  Note that this equation is damped by the 
same amount as the driven equation~\eqref{eq:driven}.\\
This is a second-order linear equation
of constant coefficients, so it can be analytically solved and we get the general solution
(for the underdamped\footnote{The other cases (critical damping and overdamping) are treated similarly.} case $\beta <\omega_0$) in the form
$$
f(t)={e}^{-\frac{\beta\,t}{2}}\,\left( k_1\,\mathrm{sin}\left( \frac{\sqrt{4\,{\omega}_{0}^{2}-{\beta}^{2}}\,t}{2}\right) +k_2\,\mathrm{cos}\left( \frac{\sqrt{4\,{\omega}_{0}^{2}-{\beta}^{2}}\,t}{2}\right) \right) ,
$$
where $k_1,k_2$ are constants of integration. The solution verifying $f(0)=0$ can be written as
$$
f(t)=C\frac{2}{\sqrt{4\,{\omega}_{0}^{2}-{\beta}^{2}}}\,{e}^{-\frac{\beta\,t}{2}}\,\mathrm{sin}\left( \frac{\sqrt{4\,{\omega}_{0}^{2}-{\beta}^{2}}\,t}{2}\right),
$$
from which we clearly see that its zeros are located at the values $t=\frac{2k\pi}{\sqrt{4\,{\omega}_{0}^{2}-{\beta}^{2}}}$,
for $k\in\mathbb{Z}$. The first zero after $t=0$ is located at $t=\frac{2\pi}{\sqrt{4\,{\omega}_{0}^{2}-{\beta}^{2}}}$, so
the solution \eqref{soldamped} is a minimum on the interval $[0,\frac{2\pi}{\sqrt{4\,{\omega}_{0}^{2}-{\beta}^{2}}}[$. When $\beta =0$ this
interval becomes $[0,\frac{\pi}{\omega_0}[$, which is the particular case of the harmonic oscillator (see \cite{gray}, pg. 446).

\subsection{Lane-Emden Equations}

The Lane-Emden second-order differential equation was originally proposed by Lane~\cite{lane}, and 
studied in detail by Emden~\cite{emden} and Fowler~\cite{fowler}, 
in order to understand  
equilibrium configurations of spherical clouds of gas 
(self-gravitating polytropic gas spheres)~\cite{binney}, \cite{chandra},
\cite{collins}.  Lane-Emden equations also appears 
in several other contexts such as 
viscous fluid dynamics, radiation, 
condensed matter, relativistic mechanics, and 
even for systems under chemical reactions (see~\cite{wong}, 
~\cite{goenner}, 
and references therein for an account of its applications. For a mathematical treatment of their zeros, see \cite{coffman}).
The Lane-Emden equation is characterized by a non-linear 
term $y^n(x)$, where the non-negative parameter 
$n\in\mathbb{Z}$ (the polytropic index, in its original 
context)
defines the nature of the second-order differential equation 
\beq
\label{eq:LE}
\frac{1}{x^2} \frac{\mathrm{d}}{\mathrm{d}x} \left({x^2 \frac{\mathrm{d}y}{\mathrm{d}x}}\right) + y^n = 0 \,,
\eeq
which may be obtained from the associated Lagrangian
\beq
L(x,y,y')=x^2 \left( \frac{y'^{\, 2}}{2}-\frac{y^{n+1}}{n+1}\right)  \,,
\eeq
in the sense that its Euler-Lagrange equations reduce to~\eqref{eq:LE}. 
For this Lagrangian, the functions~\eqref{eq:Pfunc} 
and~\eqref{eq:Qfunc} are given by $P=x^2$ and $Q=-nx^2 y^{n-1}$,  respectively.  As in the preceding example, the 
function $P$ is always positive definite on any interval of the form $]0,b]$, $b\in\mathbb{R}$.  We also note
that the function $Q$ depends on the parameter $n$.  
In this way, we 
obtain the Jacobi equation~\eqref{eq:Jacobi}
\beq
\label{eq:LE-Jacobi}
\frac{1}{x^2}\frac{\mathrm{d}}{\mathrm{d}x}\left( x^2\frac{\mathrm{d}f}{\mathrm{d}x} \right)
+ny^{n-1}f=0  \,.
\eeq
As the most frequent analytical solutions to the Lane-Emden equation are 
those corresponding to the values $n=0,1,5$,~\cite{wong} (but see~\cite{goenner} for other cases),  
we will focus next 
on solution of both, Euler-Lagrange and Jacobi equations, 
for these cases.

\subsubsection{n=0}\label{lane0}
%
Consider any interval $[0,b]$.
For this case, the general solution 
$y(x)$ of Euler-Lagrange's equation~\eqref{eq:LE} 
reads
$$
y(x)=-\frac{x^2}{6}+ \frac{k_2}{x} + k_1  \,,
$$ 
where $k_1,\ k_2$ are arbitrary integration constants.
Note that this function is singular at the origin. The physical
origin of the problem demands that the solution verify $y(0)=1$,
$y'(0)=0$, so we must take $k_2 =0$, $k_1 =1$. The solution is then
$$
y(x)=-\frac{x^2}{6}+1.
$$
In this very simple case, it is not necessary to deal with the Jacobi
equation, as the corresponding Lagrangian
$$
L(x,y,y')=x^2 \left( \frac{y'^2}{2}-y \right) ,
$$
determines, for each $x\in [a,b]$ fixed, a function $L(x,\cdot ,\cdot )$
which is convex on the convex set $\mathbb{R}\times \mathbb{R}$, as it
has a semi-definite Hessian: for any $u,v\in\mathbb{R}$,
\begin{align*}
&\det \mathrm{Hess}L(x,u,v)=0 \\
&D_{22}L(x,u,v)=0 \\
&D_{33}L(x,u,v)=x^2\geq 0.
\end{align*}
Thus, the solution is minimal on $]0,b]$, for any $b>0$.

\subsubsection{n=1}\label{lanen1}
The general solution the Lane-Emden equation~\eqref{eq:LE}
is given by 
\beq\label{lane1}
y(x)=k_1\frac{\sin(x)}{x}+ k_2\frac{\cos(x)}{x} 
\eeq
where $k_1,\ k_2$ are constants.
Again, the physical meaning of the problem imposes the condition that the solutions be defined at
$x=0$ (where it must be $y(0)=1$), so we must take $k_1 =1,k_2 =0$, getting the $\mathrm{sinc}$
function
\beq\label{sinc}
y(x)=\frac{\sin (x)}{x}\, .
\eeq
As the functions $P=x^2$, $Q=-x^2$ in this case,  
Jacobi equation~\eqref{eq:LE-Jacobi} results again a 
Lane-Emden equation with $n=1$ for the function $f(x)$.
Hence, for $n=1$ solutions of the Jacobi equation and 
the Lane-Emden equation are of an identical nature.
Thus, in view of \eqref{lane1}, the solutions of the Jacobi
equation defined for $x=0$ are those of the form $f(x)=C\mathrm{sinc}(x)$.
They have zeros located at the points $x=k\pi$, $k\in \mathbb{Z}-\{1\}$. So, on
any interval of the form $](k-1)\pi ,k\pi[$, $k\in\mathbb{Z}-\{0,1\}$, the solutions \eqref{sinc} are minimal.
A minimum is also obtained on $]-\pi ,0[$ and $]0,\pi [$.

\subsubsection{n=5}\label{n5}
Analogously, in this case, the general solution to~\eqref{eq:LE}
reads (see \cite{chandra}, \cite{horedt}\footnote{Although in these references only a $1-$parameter family
of solutions is given, it is easy to trace back the missing parameter $\alpha$ from the calculations presented there
(it is fixed at certain point to make the output of an integral more manageable).}):
\beq
y(x;\alpha ,\beta )=  \sqrt{\frac{\alpha}{\frac{{\left( \alpha\,x\right) }^{2}}{\beta}+\frac{\beta}{3}}} \,.
\eeq
We further note that if we impose the boundary conditions $y(0)=1$ and 
$y'(0)=0$ (which set $\alpha =1$, $\beta =3$) our solution becomes the common one 
\beq\label{sollane}
y(x;1,3) = \frac{\sqrt{3}}{\sqrt{3 + x^2}}  \,.
\eeq
We will work on this simplified solution.  For this case,  the function $Q= -5x^2y^4$, thus yielding the 
Jacobi equation
\beq\label{Jacobin5} 
\dfrac{\mathrm{d}^2f}{\mathrm{d}x^2}+ \frac{2}{x} \dfrac{\mathrm{d}f}{\mathrm{d}x}+ \frac{45}{(3+x^2)^2} f =0 \,,
\eeq
which may be simplified to its normal form~\eqref{eq:normal} 
by the 
change $f(x)=u(x)/x$:
\beq
\label{eq:n5ujacobi}
\dfrac{\mathrm{d}^2u}{\mathrm{d}x^2} + \frac{45}{(3+x^2)^2} u =0   \,.
\eeq 
Let us apply the method outlined in Remark \ref{supermethod}. The derivative of
the general solution with respect to the parameter $\alpha$, evaluated at the
values that give the solution we are considering, is
$$
D_2 L(x;1,3) =-\frac{\sqrt{3}\,\left( {x}^{2}-3\right) }{2\,{\left( {x}^{2}+3\right) }^{\frac{3}{2}}}.
$$
It is immediate to check that this is indeed a solution of Jacobi's equation \eqref{Jacobin5}. The
derivative with respect to $\beta$ gives nothing new (a multiple of $D_2 L(x;1,3)$). By making the
change of variables stated above, we get:
$$
u(x)=-\frac{\sqrt{3}\,x\left( {x}^{2}-3\right) }{2\,{\left( {x}^{2}+3\right) }^{\frac{3}{2}}}.
$$
Note that $u(0)=0$. The first (and only) zero of $u(x)$ after $x=0$ is given by $x=\sqrt{3}$. Thus,
the solution to the Lane-Emden equation for $n=5$ \eqref{sollane}, is a minimum for $x\in [0,\sqrt{3}[$.

\subsection{Quantum gravity in one dimension}

In modern physics, spin foams
models have been introduced in order to analyse 
certain generalizations of path integrals
appearing in gauge theories.
In particular, in quantum gravity the spin foam approach has been
developed as a tool to understand dynamical issues of the theory
by the introduction of discretizations describing 
the metric properties of spacetime~\cite{baez}.
To some extent, spin foams for quantum gravity 
were motivated by a particular 
discretization of general relativity known as 
Regge calculus~\cite{williams}.  In this context, 
a discrete model for a scalar field 
representing gravity in one temporal dimension
has been studied in detail in~\cite{hamber}.  
Here we present the continuum analogue of this model.
The general action functional is:
$$
J(y)=\frac{1}{2}\int \sqrt{g(x)}\left( g^{-1}(x)y'(x)^2+\omega y^2 (x)\right) \mathrm{d}x,
$$
where $g:\mathbb{R}\to\mathbb{R}$ is a positive function which acts as the metric on $\mathbb{R}$, and
will be taken in what follows as $g(x)=\exp (x)$, for simplicity. Thus our model Lagrangian will be
$$
L(x,y,y')=\frac{1}{2}\exp ( x/2)\left( \exp (-x)(y')^2 +\omega y^2 \right) \, .
$$
The Euler-Lagrange equations readily follow:
$$
\omega \exp(x/2)y+\frac{1}{2}\exp(-x/2)y'-\exp (-x/2)y''=0,
$$
or, as $\exp (-x/2)>0$,
\beq\label{elqg}
y''-\frac{1}{2}y'-\omega e^x y=0.
\eeq
By making the change of variable $u=\exp (x/2)$, we can put \eqref{elqg} in the form
$$
u^2 \left( \frac{1}{4} \frac{\mathrm{d}y}{\mathrm{d}u} -\omega y(u) \right) =0,
$$ 
which, as $u=\exp (x/2)>0$, reduces to
$$
\frac{\mathrm{d}y}{\mathrm{d}u} -4\omega y(u)=0.
$$
This equation is integrated by elementary techniques; its solutions are:
$$
y(u)=k_1 \exp (2\sqrt{\omega}u)+k_2 \exp (-2\sqrt{\omega}u),
$$
(with $k_1,k_2$ constants of integration) or, in the original variable $x$:
\beq\label{solqg}
y(x)=k_1 \exp (2\sqrt{\omega}\exp (x/2))+k_2 \exp (-2\sqrt{\omega}\exp (x/2))\, .
\eeq
The coefficients of the Jacobi equation are $P=\exp (-x/2)>0$ and $Q=\omega\exp (x/2)$, so
the Jacobi equation is (after simplifying an $\exp (x/2)>0$ factor):
$$
\frac{\mathrm{d}^2f}{\mathrm{d}x^2}-\frac{1}{2}\frac{\mathrm{d}f}{\mathrm{d}x}-\omega e^x f=0,
$$
which has the same form as the Euler-Lagrange equation \eqref{elqg} (a phenomenon already
encountered in the case of the Lane-Emden equation for $n=1$, recall Example \ref{lanen1}).
Thus, the general solution to the Jacobi equation is
$$
f(x)=c_1 \exp (2\sqrt{\omega}\exp (x/2))+c_2 \exp (-2\sqrt{\omega}\exp (x/2)),
$$
with $c_1,c_2$ constants of integration which can be fixed by the initial conditions
$f(0)=0$, $f'(0)=1$, giving
$$
c_1=\frac{e^{-2\sqrt{\omega}}}{2\sqrt{\omega}}, \quad c_2=-\frac{e^{2\sqrt{\omega}}}{2\sqrt{\omega}}.
$$
Substituting in the expression for $f(x)$ above, the solution to the Jacobi equation
can be written as
\begin{align*}
f(x) &= \frac{1}{2\sqrt{\omega}}\left( \exp ((e^{x/2}-1)2\sqrt{\omega})- \exp (-(e^{x/2}-1)2\sqrt{\omega})\right) \\
&= \frac{1}{\sqrt{\omega}}\sinh (2\sqrt{\omega} (e^{x/2}-1) )\, ,
\end{align*}
which clearly shows that there are no conjugate points to $x=0$. The solutions \eqref{solqg} are
thus minimals on their interval of definition.

\subsection{Square root Hamiltonian with a dissipation term}

Square-root Hamiltonians (or Lagrangians) 
are a standard feature of many reparametrization 
invariant field theories~\cite{carlini}, \cite{puzio}.
The action of a relativistic particle~\cite{fiziev}, \cite{louko}, \cite{lucha}, 
and the action of the Nambu string
are familiar examples~\cite{fairlie}, \cite{zwiebach}.
Further examples of physical theories where this sort of 
Hamiltonians appear include general relativity~\cite{louko},
\cite{menotti}, as well as certain approaches to 
quantum gravity~\cite{carlip}, \cite{menotti}, and 
also, they appear in brane motivated models~\cite{fairlie},
\cite{rojas}.  Aspects on the quantization of these kind
of Hamiltonian theories may be found in~\cite{berezin} and \cite{puzio},
to mention some.  In this section, we will study the 
Hamiltonian for a free particle under relativistic motion 
with a linear dissipation term, as proposed in~\cite{gzlz}.
This Hamiltonian reads
\beq
\label{eq:sqrt}
H(p,x,t):=e^{\gamma t} \sqrt{1+p^2 e^{-2 \gamma t}} \,,
\eeq
where $p$ stands for the canonical momentum associated to 
the dependent variable $x=x(t)$, and $\gamma$ is the dissipation term.  As discussed in~\cite{gzlz},
in the low velocity regime, this Hamiltonian reduces to 
the Caldirola-Kanai Hamiltonian
which describes the motion of a non-relativistic particle with a linear dissipation term~\cite{razavy}.
The Lagrangian associated to~\eqref{eq:sqrt} is given by
\beq
L(t,x,\dot{x})=-e^{\gamma t} \sqrt{1-(\dot{x})^2}  \,,
\eeq
and thus the solution to Euler-Lagrange equation 
\beq
\frac{\mathrm{d}}{\mathrm{d}t} \left(  \frac{e^{\gamma t}\dot{x}}{\sqrt{1-(\dot{x})^2}}\right)=0  
\eeq
is given by 
\beq\label{solsqroot}
x(t)=-\frac{A_0\,\mathrm{arcsinh}\left( {e}^{-\gamma\,t}\,\left| A_0\right| \right) }{\gamma\,\left| A_0\right| }  \,,
\eeq
where $A_0$ is an integration constant.  The 
functions~\eqref{eq:Pfunc} and~\eqref{eq:Qfunc} are
$P= e^{\gamma t}/ (1-(\dot{x})^2)^{3/2}$ and $Q=0$, 
respectively, and thus the Jacobi equation yields
\beq
\label{eq:Jsqrt}
\dfrac{\mathrm{d}}{\mathrm{d}t} \left( \frac{e^{\gamma t}}{(1-(\dot{x})^2)^{3/2}} \dfrac{\mathrm{d}f}{\mathrm{d}t}   \right)= 0 \,.
\eeq
In this last equation, the term $\dot{x}$ must be 
understood as the time-depending function 
$\dot{x}(t)= A_0 e^{-\gamma t}/\sqrt{1+A_0^2e^{-2 \gamma t }}$.
Therefore, for this model we are able to explicitly 
find the solution to Jacobi equation~\eqref{eq:Jsqrt}
\beq
f(t)=c_0- \frac{c_1}{\gamma} \frac{e^{-\gamma t}}{\sqrt{1+A_0^2e^{-2 \gamma t }}}  \,,
\eeq
being $c_0$ and $c_1$ integration constants.  
We then note that this solution has a unique zero 
at the value 
$t=(1/\gamma) \log 
\left( \sqrt{c_1^2 - A_0^2 \gamma^2 c_0^2}/ 
\gamma c_0\right)$, so there are no conjugate points
for the function $f(t)$, and the solution \eqref{solsqroot}
is a minimum for the action on any interval $[0,b]$, $b\in \mathbb{R}$.

\subsection{Quartic potential model}\label{qpotential}

In this section we will develop 
a model inspired by the static 
kink of the well-known $\phi^4$ model
in quantum field theory~\cite{rama}, \cite{pichugin}.
The model can be resolved both on 
classical and quantum grounds, and 
contains soliton solutions (see below). In 
the context of brane theories, the so-called
kink model also appears by the inclusion of 
rigidity terms associated to the intrinsic curvature
in their effective actions~\cite{zlosh}, \cite{nersesyan}.\\
The Lagrangian for our model is
\beq
L(t,x,\dot{x})= \frac{1}{2}(\dot{x})^2 + \frac{\lambda}{4} \left(x^2 - \frac{m^2}{\lambda} \right)^2  \,,
\eeq
where $m$ and $ \lambda>0$ are arbitrary real constants.
The Euler-Lagrange equation associated to this Lagrangian
reads
\beq\label{elquartic}
\ddot{x}- \lambda x \left(x^2 - 
\frac{m^2}{\lambda} \right) = 0 \,.
\eeq
This equation is of L\'enard type: $\ddot{x}+f(x)\dot{x}+g(x)=0$, where $f\equiv 0$ and
$g(x)=- \lambda x \left(x^2 - \frac{m^2}{\lambda} \right)$. The change of variables $u(x)=\dot{x}$
converts it in the first-order equation $uu'=\lambda x \left(x^2 - \frac{m^2}{\lambda} \right)$,
which is immediately integrated to give
\beq\label{integral}
\int \frac{\mathrm{d}x}{\sqrt{\lambda x^4 -2m^2 x^2 -2b}}=\frac{1}{\sqrt{2}}\int \mathrm{d}t=\frac{t-a}{\sqrt{2}},
\eeq
where $a,b\in\mathbb{R}$ are integration constants. The solutions commonly found in the
literature (cited above) are obtained by taking $b=-\frac{m^4}{2\lambda}$, so to get a perfect
square in the radicand of \eqref{integral}. In this way, the resulting solutions are:
$$
y(t)=\pm\frac{m}{\sqrt{\lambda}}\tanh \left( \frac{m(t-a)}{\sqrt{2}} \right)  \,.
$$
The solution 
with the plus sign is commonly termed the kink solution, 
while the one with the minus sign is called the anti-kink
solution. Both solutions are bounded by the values 
$\pm m/\sqrt{\lambda}$. In particular, 
the energy density of the kink solution goes as
the fourth power of the hyperbolic secant, and 
is localised within a width characterised by 
the quantity $l/m$~\cite{rama}.
However, other solutions exist. For instance, we could as well take $b=0$ in \eqref{integral}
to get (through an obvious change of variable):
$$
\frac{t-a}{\sqrt{2}}=\int \frac{\mathrm{d}x}{\sqrt{\lambda x^4 -2m^2 x^2}}=\frac{1}{m\sqrt{2}}
\int \frac{\mathrm{d}\eta}{\eta \sqrt{\eta^2 -1}}=\frac{1}{m\sqrt{2}}\arcsin \eta \, ,
$$
and hence the solution to the Euler-Lagrange equation \eqref{elquartic}:
\beq\label{solquartic}
x(t)=m\sqrt{\frac{2}{\lambda}}\sec (m(t-a)).
\eeq
In order to get the Jacobi equation~\eqref{eq:Jacobi},
we consider the functions $P=1>0$ and $Q=3 \lambda x^2 - m^2$,
which set the equation for the function $f(t)$:
\beq
\label{eq:Jkink}
\dfrac{\mathrm{d}^2f}{\mathrm{d}t^2}+m^2 \left( 1-6\sec^2 (m(t-a)) \right)f=0  \,.
\eeq
We can take $a=0$ and $m=1$ without loss of generality (these are just re-scalings). Then,
the Jacobi equation has the form $\ddot{f}+\phi (t)f=0$, where
$\phi (t)=1-\sec^2 t\leq 0$ in the interval $]-\pi/2,\pi/2[$. At the points $t=\pm \pi/2$,
the solutions have a blow-up and are not defined (so they can not be extended beyond
these points). Thus, the solutions to the Jacobi equation are defined on $]-\pi/2,\pi/2[$,
do not possess conjugate points in this interval (see Theorem \ref{nozeros}) and the
solution \eqref{eq:Jkink} is a true minimum on $]-\pi/2,\pi/2[$.

\subsection{Probability density and maximal entropy}\label{entropia}

In this section we implement a constrained 
Lagrangian system related to a probability density
function~\cite{stir}.  
In Bayesian probability theory and in statistical 
mechanics, this system is related to the 
principle of maximum entropy \cite{jaynes}, which also appears in other branches of physics, and in chemistry
and biology~\cite{buck}, \cite{klein}, \cite{martyushev}, \cite{wang}. 
The model is defined as follows.  Let $Z$ 
be a random variable, and let $\rho(x)$ 
its associated density function, so $\rho :\mathbb{R}\to ]0,+\infty[$.
Suppose that we know the second order momentum
\beq
\sigma^2 = \int_{\mathbb{R}}x^2 \rho(x)\mathrm{d}x  \,,
\eeq
and that we want to obtain the least 
biased density function $\rho(x)$.  
This may be written as the problem of finding 
the maximals for the entropy functional (defined in terms of the information theory):
\beq
S(\rho )=- \int_{\mathbb{R}} \rho (x)\log\rho (x)\mathrm{d}x,
\eeq 
and subject to the 
constraints
\beq
\label{eq:entrocons1}
\int_{\mathbb{R}} \rho(x) \mathrm{d}x = 1  \,,\\
\label{eq:entrocons2}
\int_{\mathbb{R}} x^2 \rho(x) \mathrm{d}x = \sigma^2  \,.
\eeq
Notice that the Lagrangian here, is defined on $U=\mathbb{R}\times ]0,+\infty [\times\mathbb{R}$,
although its dependence on the first and third variables is trivial.\\
Thus, we have the three functionals (in the notation of subsection \ref{sec:constraints})
\begin{align*}
&S(\rho )=- \int_{\mathbb{R}} \rho (x)\log\rho (x)\mathrm{d}x, \\
&K_1 (\rho )=\int_{\mathbb{R}} \rho(x) \mathrm{d}x ,\\
&K_2 (\rho )=\int_{\mathbb{R}} x^2 \rho(x) \mathrm{d}x .
\end{align*}
It is immediate to compute the variations:
\begin{align*}
&\delta S(\rho ,h)=- \int_{\mathbb{R}} h(x)(1+\log\rho (x))\mathrm{d}x, \\
&\delta K_1 (\rho ,h)=\int_{\mathbb{R}} h(x) \mathrm{d}x ,\\
&\delta K_2 (\rho ,h)=\int_{\mathbb{R}} x^2 h(x) \mathrm{d}x ,
\end{align*}
so it is obvious that they are weakly continuous. Let us apply the theorem \ref{th:multipliers} on
Lagrange multipliers. The case $\det (\delta K_i (y,h_l ))=0$ ($1\leq i,l\leq 2$), would lead to
$$
\det \left(
\begin{array}{cc}
\int_{\mathbb{R}} h_1(x) \mathrm{d}x & \int_{\mathbb{R}} h_2(x) \mathrm{d}x \\ 
\int_{\mathbb{R}} x^2 h_1(x) \mathrm{d}x & \int_{\mathbb{R}} x^2 h_2(x) \mathrm{d}x \\ 
\end{array} 
\right)=0
$$
for arbitrary $h_1 ,h_2 \in X$, or:
$$
\frac{\int_{\mathbb{R}} x^2 h_1(x) \mathrm{d}x}{\int_{\mathbb{R}} h_1(x)}=
\frac{\int_{\mathbb{R}} x^2 h_2(x) \mathrm{d}x}{\int_{\mathbb{R}} h_2(x)},
$$
which is absurd. Thus, we can introduce two Lagrange multipliers $\lambda_1 ,\lambda_2 $ and consider
the Lagrangian
$$
L(x,\rho ,\rho')=-\rho (x) \log\rho (x) + \lambda_1 \rho (x)+\lambda_2 x^2 \rho (x).
$$
The Euler-Lagrange equation yields
$$ 
-\log \rho (x) - 1 + \lambda_1 + \lambda_2 x^2 = 0 
$$
with solution $\rho(x) = e^{-1 + \lambda_1 + \lambda_2 x^2}$.
Substitution of this solution into the 
constraints~\eqref{eq:entrocons1} and~\eqref{eq:entrocons2}
implies that the Lagrange multipliers 
are equal to $\lambda_1 = 1 +\log\frac{1}{\sqrt{2 \pi} \sigma}$ and $\lambda_2 = -1/(2 \sigma^2)$, 
respectively. Thus, the solution of the 
Euler-Lagrange equation reads
\beq\label{rho}
\rho(x) = (\sqrt{2 \pi}\, \sigma)^{-1}\exp\left(-x^2/(2 \sigma^2)\right).
\eeq

\noindent Finally, we see that the original Lagrangian  
$L_0(x,\rho ,\rho')=-\rho (x)\log\rho (x)$, can
actually be seen as a real function of a single variable
on $]0,+\infty [$, for which the second derivative
$L''_0 (\rho )=-1/\rho $ is always negative. Therefore,
$L_0$ is concave and, due to Theorem \ref{thm:convex} and Remark \ref{concave},
the solution obtained is a (global) maximal in $]0,+\infty [$.

\begin{rem}
\emph{A posteriori}, we see that the solution we have found, \eqref{rho},
belongs to $\mathcal{C}^1(\mathbb{R})$. However, this is not obvious \emph{a priori}. Indeed, the method of Lagrange
multipliers is \emph{not} the best one to deal with the problem involving higher order moments, precisely because the
eventual solution may lie outside the space from which we start, see \cite{harremoes}. 
\end{rem}

\textbf{Acknowledgements:} The research oh the third author (JAV) was partially supported by the Mexican
Consejo Nacional de Ciencia y Tecnolog\'ia CONACyT Project CB-2012 179115.

\end{document}